\documentclass[11pt]{article}
\usepackage{amsmath,amsfonts,amsthm,amssymb,bbm,mathrsfs,enumerate,nicefrac}
\usepackage{graphicx, xcolor}
    \usepackage[driverfallback=hypertex,pagebackref=true,colorlinks]{hyperref}
    \hypersetup{linkcolor=[rgb]{.7,0,0}}
    \hypersetup{citecolor=[rgb]{0,.7,0}}
    \hypersetup{urlcolor=[rgb]{.7,0,.7}}
\usepackage{geometry}
\usepackage{epstopdf}
\usepackage[plain,noend]{algorithm2e}
\AtBeginDocument{%
  \addtolength\abovedisplayskip{-0.25\baselineskip}%
  \addtolength\belowdisplayskip{-0.15\baselineskip}%
}

\usepackage[titletoc,title]{appendix}
\usepackage[capitalise]{cleveref}
\usepackage{url}

\makeatletter
\newcommand{\mt}{\@ifnextchar^{}{^{}}}

\SetAlCapSkip{1em}
\SetKw{kwLet}{let}

\newcommand{\poly}{\mathrm{poly}}
\newcommand{\polylog}{\mathrm{polylog}}
\newcommand{\supp}{\mathrm{supp}}

\newtheorem{theorem}{Theorem}[section]
\newtheorem{lemma}[theorem]{Lemma}
\newtheorem{proposition}[theorem]{Proposition}

\newtheorem{problem}{Open Problem}

\theoremstyle{remark}
\newtheorem*{remark}{Remark}
\theoremstyle{definition}
\newtheorem{definition}{Definition}

\title{Randomized vs. Deterministic Separation in Time-Space Tradeoffs of Multi-Output Functions}

\author{Huacheng Yu\thanks{Supported by Simons Junior Faculty Award - AWD1007164.}\\ Princeton University \\ \texttt{yuhch123@gmail.com}
\and Wei Zhan\thanks{Supported by a Simons Investigator Award and by the National Science Foundation grant No. CCF-2007462.}\\ Princeton University \\ \texttt{weizhan@cs.princeton.edu}}

\date{}

\begin{document}
\maketitle
	
	\begin{abstract}
		We prove the first polynomial separation between randomized and deterministic time-space tradeoffs of multi-output functions. In particular, we present a total function that on the input of $n$ elements in $[n]$, outputs $O(n)$ elements, such that:
		\begin{itemize}
			\item There exists a randomized oblivious algorithm with space $O(\log n)$, time $O(n\log n)$ and one-way access to randomness, that computes the function with probability $1-O(1/n)$;
			\item Any deterministic oblivious branching program with space $S$ and time $T$ that computes the function must satisfy $T^2S\geq\Omega(n^{2.5}/\log n)$.
		\end{itemize}
		This implies that logspace randomized algorithms for multi-output functions cannot be black-box derandomized without an $\widetilde{\Omega}(n^{1/4})$ overhead in time.
		
		Since previously all the polynomial time-space tradeoffs of multi-output functions are proved via the Borodin-Cook method, which is a probabilistic method that inherently gives the same lower bound for randomized and deterministic branching programs, our lower bound proof is intrinsically different from previous works.
		
		We also examine other natural candidates for proving such separations, and show that any polynomial separation for these problems would resolve the long-standing open problem of proving $n^{1+\Omega(1)}$ time lower bound for decision problems with $\polylog(n)$ space.
	\end{abstract}
	
	\section{Introduction}
	
	Time-space tradeoff is the phenomenon in computation where one could trade time for space by recomputing the intermediate results instead of storing them. Apart from being ubiquitous, time-space lower bounds also help us understand what can or cannot be efficiently computed with limited memory. However, lower bounds for decision problems turned out difficult to prove, and it remains a major open problem to prove lower bounds polynomially better than trivial in the general non-uniform model of sequential computation, i.e. branching programs.
	
	By contrast, there is an abundance of lower bounds against functions that have multiple bits of output (polynomial in the length of the input). The study of time-space tradeoffs of multi-output functions started with Borodin and Cook \cite{BC82} who showed that sorting $n$ numbers within space $S$ and time $T$ requires $TS=\Omega(n^2)$. Their problem formulation and proof was later refined by Beame \cite{Bea91}, who also showed a similar lower bound for listing unique elements. Tight time-space lower bounds was also proved for a variety of multi-output problems, including algebraic problems like matrix multiplication and inversion \cite{Abr91}, frequency moments over sliding windows \cite{BCM13}, and more recently, the memory game \cite{chakrabarti2017time}, printing and counting \textsc{SAT} assignments \cite{MW18} and multiple collision finding \cite{Dinur20}, just to name a few.
	
	Despite being problems from different backgrounds and of different nature, the proofs of their lower bounds are all direct applications of the original methods of Borodin and Cook \cite{BC82}. A formal (but restrictive for certain applications) and detailed description of the Borodin-Cook method was given in \cite[Section 10.11]{Sav98}, and here we present a brief framework of the method:
	
	\begin{enumerate}
		
		\item Fix a distribution $\mathcal{D}$ over the inputs (often uniform), and find $a(S)$ such that given $a(S)$ bits in the input, only $S$ bits of output are revealed on average. 
		
		\item Prove that for some large $c>1$, given any decision tree of depth $a(S)$ and $c\cdot S$ bits of output assigned to each path in the tree, these outputs are correct with probability $2^{-\Omega(S)}$ under $\mathcal{D}$. This is usually the technical part of the proof. 
		
		\item Now split the branching program into stages of length $a(S)$, and by a union bound over the $2^S$ starting nodes of each stage, the above argument shows that most inputs under $\mathcal{D}$ cannot generate $c\cdot S$ bits output within a stage. This implies a lower bound of the form $ST/a(S)\geq m$ where $m=\poly(n)$ is the number of output bits.
	
	\end{enumerate}

	Note that here we state the method for proving lower bounds against deterministic branching programs. When the distribution $\mathcal{D}$ above is uniform, the method can be stated even more simply as a counting argument over the set of inputs, and the lower bounds hold for average-case complexity. But no matter how $\mathcal{D}$ is chosen, by Yao's Minimax Principle the above arguments actually provide lower bounds against distributions over deterministic branching programs, the most general model of randomized sequential computation. In other words, the Borodin-Cook method \emph{inherently} give the same lower bound for deterministic and randomized computation. As all the previous lower bounds are proved using the Borodin-Cook method, the following question remains unanswered before this work:
	
	\begin{quote}\emph{
		Is there a polynomial separation between randomized and deterministic branching programs for time-space tradeoffs of multi-output functions?}
	\end{quote}

	Here we answer this question in the affirmative for oblivious branching programs, where the queries made to the input in the branching programs are independent of the input. We design a problem called \textsc{$(n,p)$-Non-Occurring Elements},  which can be solved efficiently by randomized oblivious algorithms, while any deterministic oblivious branching program solving the problem requires either polynomially larger space or polynomially longer time. Our problem has the additional advantage that it is a total function, which prevents the trivial separation where a randomized algorithm may have a sublinear running time (see our discussion in \Cref{sec:rel}).
	
	\begin{definition}
		Let $n>1$ and $p$ be a prime factor of $n$. In the \textsc{$(n,p)$-Non-Occurring Elements} (\textsc{$(n,p)$-NOE} for short) problem, the input is an unordered list of $n$ numbers $X=(x_1,\ldots,x_n)\in[n]^n$. The output is a set $Y\subseteq[n]$ such that:
		\begin{itemize}
			\item  If for every $c\in[n]$, the number of times that $c$ occurs in $X$ is a multiple of $p$ ($0$ included, so there are at most $n/p$ distinct \emph{occurring} elements), then $Y$ consists of the (at least $n-n/p$) elements in $[n]$ that \emph{do not occur} in $X$;
			\item Otherwise $Y=\varnothing$.
		\end{itemize}
	\end{definition}
	
	\begin{theorem}\label{thm:main}
		There is a randomized oblivious branching program with space $O(\log n)$ and time $\max\{1,n/p^2\}\cdot O(n\log n)$, that computes \textsc{$(n,p)$-NOE} with probability at least $1-2/n$. Moreover, the algorithm can be implemented with one-way access to random bits. On the other hand, any deterministic oblivious branching program with space $S$ and time $T$ that correctly computes \textsc{$(n,p)$-NOE} must satisfy $T^2(S+\log T)\geq \Omega(n^3/p)$.
	\end{theorem}
	\Cref{thm:main} will be proved in \Cref{sect:sep}. Taking $n=p^2$, we get a polynomial separation with randomized upper bound $S=O(\log n), T=O(n\log n)$ and deterministic lower bound $T^2S\geq \widetilde{\Omega}(n^{2.5})$.
	
	\begin{remark}
	The \textsc{$(n,p)$-NOE} problem could be perceived as a ``promised'' version of the \textsc{Non-Occurring Elements} problem (in which the output at all times consists of the elements not occurring in $X$), and the latter problem has time-space tradeoff $TS=\Theta(n^2)$ for both deterministic and randomized branching programs \cite{MW18}. The promise that every elements occurs a multiple of $p$ times can be efficiently checked with randomness (\Cref{lemma:promise}), however there may as well be a deterministic algorithm that verifies the promise in almost-linear time and poly-log space (subject to \Cref{openproblem} below). The above facts imply that neither the \textsc{Non-Occurring Elements} problem nor the promise itself could demonstrate the desired separation.
	\end{remark}
	
	\subsection{Separations with Implications on Decision Problems}
	
	We shall stress that our proof of \Cref{thm:main} is not technically hard. Indeed, to bypass the inherent disadvantage of the Borodin-Cook method, the key in our lower bound proof of \Cref{thm:main} is to choose \emph{adversarially} a distribution $\mathcal{D}$ that \emph{depends on} the deterministic structure of the branching program instead of fixing a distribution $\mathcal{D}$ in advance. Arguably, the difficulty in proving a separation lies mostly in finding a proper total function where the adversarial method works. We demonstrate this difficulty by showing that, for several natural candidate problems whose best known deterministic algorithms are polynomially worse than randomized algorithms, proving a polynomial separation will lead to the resolution of the following open problem:
	
	\begin{problem}\label{openproblem}
		Find an explicit family of decision problems $F:\{0,1\}^n\rightarrow\{0,1\}$, such that any deterministic branching program with space $S\leq\polylog(n)$ that computes $F$ requires time $T=n^{1+\Omega(1)}$.
	\end{problem}

	As we mentioned at the beginning of this paper, \Cref{openproblem} is a long-standing and notoriously hard problem. Even against oblivious branching programs, the best time-space lower bound is still $T=\Omega(n\log^2(n/S))$ proved by Babai, Nisan and Szegedy \cite{BNS92} three decades ago.
	
	For a concrete example of our implication results, one of the candidate functions that we study can be very succinctly described as follows:
	\begin{definition}
		In the \textsc{2-StepPointerChasing} (\textsc{2-PC} for short) problem, the input is a function $f:[n]\rightarrow[n]$, and the output consists of $(x,f(f(x)))$ for all $x\in[n]$.
	\end{definition}
	The \textsc{2-PC} problem exhibits interesting phase transitions in time-space tradeoffs with different computation models. With non-oblivious queries, \textsc{2-PC} can be solved deterministically in space $O(\log n)$ and time $O(n)$. On randomized oblivious branching programs, \textsc{2-PC} obliges to the tradeoff $T^2S=\widetilde{\Theta}(n^3)$. We conjecture that randomness is required for this tradeoff, and on deterministic oblivious branching programs the lower bound $TS\geq\widetilde{\Omega}(n^2)$ holds. However, it turns out that proving such a separation (in fact any lower bound polynomially better than $T^2S=\Omega(n^3)$) is at least as hard as answering \Cref{openproblem}. We will show this in \Cref{thm:2pc}, by relating \textsc{2-PC} to a matching problem on explicit bipartite expanders.
	
	Another example is the \textsc{SetIntersection} problem (given two sets $A$ and $B$, output elements in $A\cap B$). The optimal randomized algorithm for \textsc{SetIntersection} is based on the small-space collision finding algorithm \cite{BCM13} on random (or pseudo-random \cite{CJWW22,LZ23}) hash functions, with time-space tradeoff $T^2S=\widetilde{O}(n^3)$. The randomness seems essential in the original algorithm and its pseudo-random improvements; however, in \Cref{thm:ed} we give a black-box reduction from \textsc{SetIntersection} to the well-studied decision problem of \textsc{ElementDistinctness}, which shows that proving a polynomial separation for \textsc{SetIntersection} would answer \Cref{openproblem} on \textsc{ElementDistinctness}.
	
	\subsection{Related Works}\label{sec:rel}
	
	\paragraph{Query Complexity.}
	In essence, time-space tradeoff in non-uniform models is a study on query complexity with bounded memory. There are extensive results on separations for query complexity in different models, and readers can refer to \cite{ABKRT21} for a comprehensive lists of separations and relations between query complexities and other complexity measures on total functions.
	
	Our result is in a parallel world to these separations: The separations in query complexity are based on \emph{sub-linear} query algorithms, while we demonstrate our separation with \emph{super-linear} lower bounds. In fact, for oblivious queries, there is no separation in query complexity: It was shown in \cite{Montanaro10} that no matter deterministic, randomized or even quantum, the oblivious (or non-adaptive) query complexity for \emph{any} total boolean function that depends on $n$ variables is always $\Omega(n)$. This is the major reason that we focus on total functions, so that our separation is substantial in the space-bounded setting.
	
	\paragraph{Superfast Derandomization.}
	With the common belief that $\mathsf{BPP}=\mathsf{P}$ and $\mathsf{BPL}=\mathsf{L}$ and the long line of literature under the ``Hardness vs. Randomness'' paradigm, several recent works take interest in the question of how fast derandomization could be. In particular, Chen and Tell \cite{CT21} showed that under standard hardness assumption, linear overhead in the running time is possible, which is optimal assuming \texttt{\#NSETH}.
	
	In the space-bounded setting, Hoza \cite{Hoza19} proved that any decision problem with space $S$, time $n\cdot\poly(S)$ and one-way access to randomness can be solved with the same space and time bound and only $O(S)$ random bits. It is not clear whether a similar result holds for full derandomization even for $S=O(\log n)$, while our \Cref{thm:main} can be viewed as an impossibility result for multi-output functions, that a logspace randomized algorithm cannot be black-box derandomized (so that it keeps the oblivious query pattern) without an $\widetilde{\Omega}(n^{1/4})$ time overhead. We conjecture that the separation holds for general non-oblivious branching programs, so that the black-box assumption could be removed, and that the separation can be improved to show an \emph{unconditional} $\widetilde{\Omega}(n)$ overhead.
	
	\paragraph{Derandomizing Element Distinctness.}
	In \cite{BCM13}, Beame, Clifford and Machmouchi presented an randomized algorithm that solves \textsc{ElementDistinctness} (\textsc{ED} for short, that decides whether $n$ input elements are all distinct) with tradeoff $T^2S=\widetilde{O}(n^3)$, which is strictly better than sorting. The algorithm requires access to a large random hash function that does not count towards the space usage. Recent works \cite{CJWW22,LZ23} managed to bring down the seed length of the hash function to $O(\log^3 n)$ and thus acquiring an algorithm with the same tradeoff but also one-way access to random bits. They raised the question of how small the random seed length could be. At an extreme, one may wonder if the algorithm can be even fully derandomized, so that \textsc{ED} is solved deterministically with the same tradeoff. 
	
	This question cannot be answered in negative without answering \Cref{openproblem}. However, there are numerous multi-output functions, such as \textsc{SetIntersection}, that uses the same algorithm to achieve the $T^2S=\widetilde{O}(n^3)$ tradeoff, and proving stronger deterministic lower bounds would also imply that the algorithm of \cite{BCM13} cannot be fully derandomized without polynomial overhead. Yet in \Cref{sect:ed}, we show that for all such problems for which we have a tight randomized lower bound, this route is also impossible without answering \Cref{openproblem}.
	
	\paragraph{Quantum Time-Space Tradeoffs}
	Finally, we would like to mention that the situation of time-space tradeoffs of multi-output functions is entirely different in quantum computing. Klauck \cite{Klauck03} presented a quantum algorithm for sorting $n$ numbers with $T^2S=\widetilde{O}(n^3)$, which already provides a polynomial separation between quantum and randomized time-space tradeoff compared to the classical lower bound in \cite{BC82,Bea91}. The separation is extra strong in the sense that the quantum algorithm uses only $\polylog(n)$ quantum memory.
	
	On the other hand, lower bounds for quantum time-space tradeoffs are more scarce and the proofs for Step 2 in the quantum analogy of the Borodin-Cook methods are more ad-hoc. Currently, only two proof methods are known for quantum lower bounds: via direct-product theorems \cite{KSdW07} and via the recording query technique \cite{HM21}.
	
	\section{Preliminaries}
	
	We use $[n]$ to denote the set ${1,2,\ldots,n}$. We use asymptotic notations $\widetilde{O}$ and $\widetilde{\Omega}$ to hide poly-logarithmic factors in $O(\cdot)$ and $\Omega(\cdot)$. Since throughout this paper, the input size of interest is always polynomial in $n$, they always hide factors of $\polylog(n)$ regardless of the content in the parenthesis: For instance, $\widetilde{O}(1)$ in this paper always stands for $O(\polylog(n))$.
	
	\paragraph{Branching Programs.}
	The computation models we consider in this paper are branching programs, which are general enough to model any computation with a proper notion of time and space. However, most of our upper bound results are uniform and can be implemented on more restrictive models such as RAMs.
	
	A deterministic branching program is a layered DAG, with a unique initial vertex. When the inputs $(x_1,\ldots,x_n)$ of the problem are from the domain $D$, each vertex not in the last layer is labeled with $i\in[n]$, and has $|D|$ edges going out towards the next layer labeled with elements in $D$. The \emph{computation path} on input $(x_1,\ldots,x_n)$ is the unique path that starts from the initial vertex, on each vertex labeled with $i$ queries $x_i$, and then follows the outgoing edge labeled with the value of $x_i$, until reaching the last layer. We say that a branching program is with \emph{space} $S$ and \emph{time} $T$ if it consists of at most $T+1$ layers, and each contains at most $2^S$ vertices. The branching program is \emph{oblivious}, if within each layer the labels on all the vertices are the same.
	
	A randomized (oblivious) branching program with space $S$ and time $T$ is a distribution over deterministic (oblivious) branching programs with the same space and time bound. We say that a randomized branching program has \emph{one-way access} to random bits, if in each layer the random labels on the vertices and outgoing edges are independent of the rest of the branching program.
	
	\paragraph{Multi-Output Functions.}
	For computing a decision problem, each node in the last layer of the branching program is labeled with $0$ or $1$, and the output is the label on the final node of the computation path. For computing a multi-output function $F:D^n\rightarrow R^m$ with the range set $R$, however, we allow outputting during the computation path: Each edge in the branching program is allowed to output arbitrarily many output statements $(j,y_j)\in[m]\times R$ that claims $F(x)_j=y_j$.  The final output is the vector in $R^m$ decided by the collection of output statements on the computation path, which must be complete and consistent. In many cases, the multi-output function computes a subset of $R$ instead of a vector. In such cases, the output statements are simply elements $y\in R$, and the final output is the collection of these elements.
	
	Finally, we note that as $|D|,|R|$ and $m$ are all $\poly(n)$ in this paper, all the problems we consider can be converted into the binary domain and binary range, with only $\polylog(n)$ overhead in time and space for the branching programs computing them, which are ignored as we are focused on polynomial separations.
	
	\section{Polynomial Separation for Oblivious Computation}\label{sect:sep}
	
	In this section we prove \Cref{thm:main}. The randomized upper bound is proved in \Cref{sect:rand} and the deterministic lower bound is proved in \Cref{sect:detm}.
	
	\subsection{Randomized Oblivious Upper Bound}\label{sect:rand}
		
	\begin{lemma}\label{lemma:promise}
		There is a randomized algorithm using $O(\log n)$ space and $O(\log n)$ random bits that reads $X=(x_1,\ldots,x_n)\in[n]^n$ as a one-pass stream and satisfies that:
		\begin{itemize}
			\item If every $c\in[n]$ occurs in $X$ a multiple of $p$ times, the algorithm always accepts;
			\item Otherwise, the algorithm rejects with probability at least $1-2p^{-1/2}\log n$.
		\end{itemize}
	\end{lemma}
	\begin{proof}
		The algorithm maintains a linear sketch of the frequencies of elements in $[n]$. Specifically, let $\alpha_1,\ldots,\alpha_n$ be uniformly and independently drawn from $\mathbb{F}_p$. The algorithm computes $\sum_i \alpha_{x_i}$ and accepts if the sum equals $0$. If some $c\in[n]$ occurs not a multiple of $p$ times, the factor before $\alpha_c$ in the sum is non-zero, and the sum equals $0$ with probability $1/p$.
		
		To reduce the random bit usage (the naive approach uses $n\log p$ random bits) we use Reed-Muller codes. Instead of drawing $\alpha_1,\ldots,\alpha_n$ independently, the algorithm draws $\beta_1,\ldots,\beta_m\in\mathbb{F}_p$ uniformly and independently, and let $\alpha_1,\ldots,\alpha_n$ be the values of monomials
		\[\beta_1^{d_1}\beta_2^{d_2}\cdots\beta_m^{d_m},\quad 0\leq d_1,\ldots,d_m<d.\]
		By taking $d=p^{1/2}$ and $m=2\log n/\log p$, the number of such monomials is at least $d^m\geq n$. Since $m\log p=O(\log n)$, the algorithm can draw and store $\beta_1,\ldots,\beta_m$ directly. After reading $x_i=c\in[n]$, $(c-1)$ is decomposed in base $d$ to obtain $d_1,\ldots,d_m$ in sequence, while the algorithm computes $\alpha_c=\beta_1^{d_1}\beta_2^{d_2}\cdots\beta_m^{d_m}$ and accumulates it to the sum $\sum_i \alpha_{x_i}$.
		
		Now the sum $\sum_i \alpha_{x_i}$ is a total degree $md$ polynomial in $\mathbb{F}_p$ on variables $\beta_1,\ldots,\beta_m$, where the the coefficients are the frequencies of elements in $[n]$ occurring in $X$. If every $c\in[n]$ occurs in $X$ a multiple of $p$ times, the polynomial is always zero; Otherwise, the polynomial is non-zero, and by the Schwartz-Zippel Lemma, the probability that the polynomial evaluates to zero is at most $md/p\leq 2p^{-1/2}\log n$.
	\end{proof}
	
	\begin{lemma}\label{lemma:noe}
		Suppose $X=(x_1,\ldots,x_n)\in[n]^n$ satisfies that every $c\in[n]$ occurs in $X$ either $0$ or at least $p$ times. Then there is a randomized oblivious algorithm using $O(\log n)$ space and $O(n^2p^{-2}\log n)$ time, with one-way access to random bits, that solves \textsc{Non-Occurring Elements} on $X$ with probability at least $1-1/n$.
	\end{lemma}
	\begin{proof}
		Let $R\subseteq[n]$ be a random multi-set of size $r=3np^{-1}\ln n$. As every occurring element occurs at least $p$ times, the probability that $\{x_i\mid i\in R\}$ does not contain all occurring elements in $X$ is at most
		\begin{equation}\label{eq:1}
			n\cdot(1-p/n)^r\leq n\cdot e^{-3\ln n}=n^{-2}.
		\end{equation}
		
		The algorithm goes for $n/p$ rounds, in each round independently samples such a multi-set $R$ of size $r$, and queries $x_i$ for $i\in R$. The algorithm also stores an number $j$, which is initialized as $0$, and in each round $j$ is updated to the next smallest occurring number
		\[j'=\min\ \{x_i>j\mid i\in R\}\cup\{n+1\}.\]
		At the end of each round, the algorithm outputs every number strictly between the pre-updated $j$ and $j'$. By \eqref{eq:1} and a union bound, with probability at least $1-1/n$, in every round $\{x_i\mid i\in R\}$ contains all occurring elements (where there are at most $n/p$ of them). In this case $j$ goes through all occurring elements in order, and thus the outputs are exactly the non-occurring ones.
		
		The overall time complexity is $rn/p=O(n^2p^{-2}\log n)$, and since elements in $R$ can be sampled sequentially to compute $j'$ and no need to be stored, the only space usage is for storing $j$ and $j'$ which is $O(\log n)$.
	\end{proof}

	Note that \Cref{lemma:promise} solves a decision problem and thus can be repeated for $O(\log n)$ times to amplify the success probability to $1-1/n$. Then combined with \Cref{lemma:noe}, we obtain the desired randomized oblivious upper bound of space $O(\log n)$ and time $O((1+n/p^2)\cdot n\log n)$.

	\subsection{Deterministic Oblivious Lower Bound}\label{sect:detm}
	\begin{lemma}
		Any deterministic oblivious branching program with space $S$ and time $T$ that correctly computes \textsc{$(n,p)$-NOE} must satisfy $T^2(S+\log T)\geq \Omega(n^3/p)$.
	\end{lemma}
	\begin{proof}
		Divide the branching program into $\ell=2T/n$ stages, each of which contains $T/\ell=n/2$ queries. We first construct a partition $\mathcal{P}$ on $[n]$ that consists of $n/p$ parts of size $p$ as follows:
		\begin{enumerate}
			\item Initially, let $\mathcal{P}=\varnothing$.
			\item For each stage $k$ of the branching program, let $Q_k$ be the set of indices queried in the stage. Arbitrarily pick $r=n^2/(4Tp)$ disjoint sets of size $p$ outside $\bigcup \{P\in\mathcal{P}\}\cup Q_k$, and add them into $\mathcal{P}$.
			\item Finally after going through all the stages, arbitrarily partition the remaining elements in $[n]$ into sets of size $p$.
		\end{enumerate}
		Notice that during Step 2, the total number of elements in $\bigcup \{P\in\mathcal{P}\}$ never exceeds
		\[r\ell p=\frac{n^2}{4Tp}\cdot\frac{2T}{n}\cdot p= n/2.\] As $|Q_k|\leq n/2$, this implies that Step 2 is always feasible. 
		
		We define a distribution $\mathcal{D}$ of $X\in[n]^n$ as follows: For every part $P\in\mathcal{P}$, uniformly and independently pick $c\in[n]$ and let $x_i=c$ for all $i\in P$. Notice that the \textsc{$(n,p)$-NOE} problem is identical to \textsc{Non-Occurring Elements} on $\supp(\mathcal{D})$, i.e., every element occurs a multiple of $p$ times. Now consider the probability
		\begin{equation}\label{eq:pr}
			\Pr_{X\sim\mathcal{D}}[\textrm{On input $X$, at least $m$ distinct elements are outputted in stage $k$}].
		\end{equation}
		Since for each input $X\in\supp(\mathcal{D})$ there are at least $n-n/p\geq n/2$ non-occurring elements, there must exist a stage $k$ such that the above probability is at least $1/\ell$ for $m=n/(2\ell)$.
		
		On the other hand, Step 2 in the construction of $\mathcal{P}$ implies that, given the query answers in stage $k$ (i.e. $x_i$ for all $i\in Q_k$), there are at least $r$ parts in $\mathcal{P}$ whose values in $X$ are still uniformly random. When the query answers are given, there are at most $2^S$ different collections of outputs in stage $k$ (dictated by the starting state of the stage), and if $m$ distinct elements are outputted and thus non-occurring, each one of the $r$ parts is consistent with these outputs with probability $1-m/n$, as the elements in these parts should be not in the output. Therefore the probability in \eqref{eq:pr} is upper bounded by
		\[2^S\cdot \left(1-\frac{m}{n}\right)^r = 
		2^S\cdot \left(1-\frac{n}{4T}\right)^\frac{n^2}{4Tp}\leq 2^S\cdot e^{-\frac{n^3}{16T^2p}}.\]
		As the probability is at least $1/\ell\geq 1/T$, we have
		\[S-\frac{\log e\cdot n^3}{16T^2p}\geq -\log T \quad\Rightarrow\quad T^2(S+\log T)\geq \Omega(n^3/p). \qedhere\]
	\end{proof}

	\section{Separations that Imply Decision Lower Bounds}
	
	In this section we present several natural candidates of multi-output function for randomized vs. deterministic separations, and show that actually proving such separations will lead to answering \Cref{openproblem}. These results can be perceived in two ways: On one hand, these are currently barrier results implying that proving separations for these natural problems is difficult, which is where the \textsc{$(n,p)$-NOE} problem in our main result stands out; On the other hand, one may hope that future developments in proving lower bounds for multi-output functions will help towards the final resolution of \Cref{openproblem}.
	
	Before getting into the concrete examples, we note that every multi-output function $F:\{0,1\}^n\rightarrow\{0,1\}^m$ can be converted to a decision problem $F':\{0,1\}^n\times[m]\rightarrow\{0,1\}$ defined as $F'(x,i)=F(x)_i$. Therefore, if $F'$ can be computed in space $O(\log n)$ and time $\widetilde{O}(n)$, then $F$ can trivially be computed in space $O(\log n)$ and time $\widetilde{O}(mn)$. Our results in this section holds non-trivially with better time complexity than $\widetilde{O}(mn)$. However, this implication is still useful as it makes decision problems and single-output functions (whose outputs are in $[n]$, or generally have length $m=\polylog(n)$) morally equivalent with respect to \Cref{openproblem}: Any lower bound for a single-output function that is polynomially better than trivial implies a corresponding lower bound for a decision problem.
	
	\subsection{Pointer Chasing and Expanders}\label{sect:2pc}
	
	Recall the definition of the \textsc{2-StepPointerChasing} problem:
	\begin{definition}
		In the \textsc{2-StepPointerChasing} (\textsc{2-PC} for short) problem, the input is a function $f:[n]\rightarrow[n]$, and the output consists of $(x,f(f(x)))$ for all $x\in[n]$.
	\end{definition}

	For non-oblivious algorithms, \textsc{2-PC} can be easily solved in deterministic space $O(\log n)$ and time $O(n)$, by querying $f$ on each $x$ and adaptively on $f(x)$. On the other hand we have the following lower bound on oblivious algorithms for \textsc{2-PC}. The proof is a direct application of Borodin-Cook method on the uniform distribution, and thus omitted.
	\begin{proposition}\label{prop:2pc}
		Any randomized oblivious branching program with space $S$ and time $T$ that solves \textsc{2-PC} must satisfy $T^2S\geq \widetilde{\Omega}(n^3)$.
	\end{proposition}
	Notice that \Cref{prop:2pc} provides an example of polynomial separation between oblivious and non-oblivious time-space tradeoffs of total functions. The bound is also tight and can be achieved via the following simple algorithm: In each round pick two random subsets $X,Y\subseteq[n]$ with $|X|=|Y|=\sqrt{nS}$. We store at most $\widetilde{O}(S)$ pairs of $(x,f(x))\in X\times Y$ by querying $f$ on $X$, and output the corresponding $(x,f(f(x)))$ by querying $f$ on $Y$. Each pair in a round is found with probability close to $S/n$, and thus $\widetilde{O}(n/S)$ rounds suffices.
	
	The above algorithm heavily relies on the fact that $Y$ is decided entirely by randomness and hardwired into the branching programs. A natural question is whether the same time-space tradeoff holds for oblivious computation with weaker notions of randomness, or even without randomness at all. In \Cref{thm:2pc} below, we show that proving impossibility results to this question will give answers to \Cref{openproblem}. We first need to introduce the single-output function, \textsc{ExpanderMatching} based on the explicit unbalanced bipartite expanders by Guruswami, Umans and Vadhan \cite{GUV09}.
	
	\begin{definition}
		A bipartite graph $\Gamma\subseteq[n]\times[m]$ is a $(k,a)$-expander if for every subset $L\subseteq [n]$ with $|L|\leq k$, the number of the neighbors of $L$ is at least $a\cdot|L|$.
	\end{definition}
	
	\begin{theorem}[\cite{GUV09}]\label{thm:guv}
		For every constant $\alpha>0$, given $n\in\mathbb{N}$ and $k\leq n$, there is an explicitly constructed bipartite graph $\Gamma_{\alpha,n,k}\subseteq[n]\times[m]$ which is a $(k,1)$-expander, with $|\Gamma_{\alpha,n,k}|=\widetilde{O}(n)$ and $m\leq \widetilde{O}(k^{1+\alpha})$.
	\end{theorem}
	
	The original result in \cite{GUV09} is stronger than stated in \Cref{thm:guv}, with the expansion factor $a$ arbitrarily close to the degree $|\Gamma_{\alpha,n,k}|/n=\polylog(n)$. For our application, we only need expansion to be no less than $1$. We use the graph to construct an explicit single-output function as follows:
	
	\begin{definition}
		The \textsc{$(\alpha,n,k)$-ExpanderMatching} problem is a function $[n]^k\times [m]\rightarrow [n]\cup\{\perp\}$, with $m$ decided by \Cref{thm:guv}. Given the input $L\in [n]^k$ and $y\in[m]$, we think of $L$ as a subset of $[n]$ with $|L|\leq k$. There exists a matching for $L$ in $\Gamma_{\alpha,n,k}$ because of the $(k,1)$-expander property, and we consider the lexicographically smallest matching $\mathcal{M}:L\rightarrow[m]$ in $\Gamma_{\alpha,n,k}$. The output of the problem is $\mathcal{M}^{-1}(y)$ if it exists, or $\perp$ if not.
	\end{definition}
	
	\begin{theorem}\label{thm:2pc}
		For every constant $\alpha>0$, if \textsc{$(\alpha,n,k)$-ExpanderMatching} can be solved by deterministic oblivious branching programs with space $\widetilde{O}(1)$ and time $\widetilde{O}(k)$, then for every $S\leq n$, there is a deterministic oblivious branching program solving \textsc{2-PC} with space $\widetilde{O}(S)$ and time $\widetilde{O}(\sqrt{n^{3+\alpha}/S})$.
	\end{theorem}
	\begin{proof}
		We partition $[n]$ into blocks $B_1\sqcup \cdots \sqcup B_{n/k}$ of size $k$, with $k$ to be optimally chosen later. The deterministic oblivious algorithm for \textsc{2-PC} consists of $n/(kS)$ stages, where in each stage we output $(x,f(f(x)))$ all $x$ in $S$ consecutive blocks. In order to do so we need to query $f$ on $f(B_i)$, but as the queries are oblivious, we instead query $f$ on the neighbors of $y$ in $\Gamma_{\alpha,n,k}$ for each $y\in[m]$. Since $|f(B_i)|\leq k$, the matching for $f(B_i)$ provides all the answers for $x\in B_i$. More concretely, the algorithm is described as \Cref{alg1}.
		
		\begin{algorithm}
			\DontPrintSemicolon
			\caption{Deterministic Oblivious Algorithm for \textsc{2-PC}}\label{alg1}
			\For{$\ell\gets 0,\ldots,n/(kS)-1$}{
				\For {$y\in[m]$}{
					\For {$i\in[S]$}{
						Apply \textsc{$(\alpha,n,k)$-ExpanderMatching} on $f(B_{i+\ell S})\in[n]^k$ and $y$; \\
						Store the answer $u_{i} \in [n]\cup\{\perp\}$.
					}
					\ForEach {$v\in[n]$ such that $(v,y)\in \Gamma_{\alpha,n,k}$}{
						Query $f(v)$; \\
						\lIf {$v=u_i$ for some $i\in[S]$}{attach $f(u_i)$ to $u_i$.}
					}
					\For {$x\in B_{\ell S+1}\sqcup\cdots\sqcup B_{(\ell+1)S}$}{
						Query $f(x)$; \\
						\lIf {$f(x)=u_i$ for some $i\in[S]$}{output $(x,f(u_i))$.}
					}
				}
			}
		\end{algorithm}
	
		To prove the correctness, it suffices to show that every $x\in[n]$ is outputted. This is guaranteed in every block $B_{i+\ell S}$, as when $y$ goes through $[m]$, every element in $f(B_{i+\ell S})$ is matched and appears as $u_i$ at some point.
		
		The space complexity is clearly $\widetilde{O}(S)$ as the bottleneck is storing $u_i$ and $f(u_i)$ for $i\in[S]$. To identify the time complexity, notice that $f$ is queried in all three inner loops. For each $\ell$ and $y$, the \textsc{$(\alpha,n,k)$-Expander Matching} algorithm makes $\widetilde{O}(kS)$ queries in total, while querying $f(x)$ for $x\in B_{\ell S+1}\sqcup\cdots\sqcup B_{(\ell+1)S}$ also takes $O(kS)$ time. Besides, for each $\ell$, querying $f(v)$ for every edge $(v,y)\in \Gamma_{\alpha,n,k}$ takes $|\Gamma_{\alpha,n,k}|=\widetilde{O}(n)$ time. Therefore the total number of oblivious queries is
		\[\frac{n}{kS}\left(m\cdot \widetilde{O}(kS)+\widetilde{O}(n)\right)=\widetilde{O}\left(k^{1+\alpha}n+\frac{n^2}{kS}\right).\]
		Taking $k=\sqrt{n/S}$, the above expression is upper bounded by $\widetilde{O}(\sqrt{n^{3+\alpha}/S})$.
	\end{proof}

	As a direct corollary of \Cref{thm:2pc}, if we managed to prove a polynomial separation between randomized and deterministic oblivious time-space tradeoffs of \textsc{2-StepPointerChasing}, it would imply a strong lower bound for \textsc{$(\alpha,n,k)$-ExpanderMatching} for some $\alpha>0$ and thus would answer \Cref{openproblem}. 
	
	\subsection{Element Distinctness and Collision Finding}\label{sect:ed}
	
	We recall the definition of the \textsc{ElementDistinctness} problem.
	
	\begin{definition}
		In the \textsc{ElementDistinctness} (\textsc{ED} for short) problem, the input is a list of $n$ elements from a fixed domain $D$, with $|D|=\poly(n)$. The output is $1$ if all elements are distinct, and $0$ otherwise.
	\end{definition}

	A randomized algorithm for \textsc{ED} with $T^2S=\widetilde{O}(n^3)$ was given in \cite{BCM13}, and it was later improved to use only one-way access to randomness in \cite{CJWW22,LZ23}. Based on the same algorithm, they also showed that the \textsc{SetIntersection} problem (given two sets $A$ and $B$ of size $n$, output $A\cap B$) can be solved in $T^2S=\widetilde{O}(n^3)$, and the tradeoff is known to be optimal \cite{Dinur20}. Different variants of this problem was also studied, such as memory games \cite{chakrabarti2017time} and $n$-collision finding \cite{Dinur20}, which share the same tight tradeoff for randomized algorithms.
	
	Here we present a general form of \textsc{SetIntersection}, that covers all the variants when two sets that each contains no duplicates are given, and show its black-box relationship with \textsc{ED}:
	
	\begin{definition}
		In the \textsc{SetCollision} problem, the input contains two sets $A,B\subseteq D$ given as unordered lists $(a_1,\ldots,a_n)$ and $(b_1,\ldots,b_n)$ that contain no duplicated elements in each list itself. The output consists of all collisions, that are triples $(i,j,x)$ such that $a_i=b_j=x$.
	\end{definition}
	
	\begin{theorem}\label{thm:ed}
		If \textsc{ED} can be solved deterministically with space $\widetilde{O}(1)$ and time $\widetilde{O}(n)$, then \textsc{SetCollision} can be solved deterministically with space $\widetilde{O}(1)$ and time $\widetilde{O}(n^{3/2})$. Furthermore, if the algorithm for \textsc{ED} is oblivious, then for every $S\leq n$, there is a deterministic (non-oblivious) algorithm that solves \textsc{SetCollision} with space $\widetilde{O}(S)$ and time $\widetilde{O}(\sqrt{n^3/S})$.
	\end{theorem}
	\begin{proof}
		
		We first present a simple divide-and-conquer algorithm $\mathcal{A}$ for solving \textsc{SetCollision}. The algorithm $\mathcal{A}(\ell,s,s')$, which find the collisions between two intervals of length $l$, is described recursively as \Cref{alg2}. For the sake of simplicity we assume that $\ell$ is a power of $2$.
		\begin{algorithm}
			\DontPrintSemicolon
			\caption{Divide-and-Conquer Algorithm $\mathcal{A}(\ell,s,s')$} \label{alg2}
			\If {$\ell=1$}{
				\lIf {$a_s=b_{s'}$}{Output $(s,s',a_s)$;}
				\Return.
			}
			\lIf {$\mathrm{ED}(a_s,\ldots,a_{s+\ell-1},b_{s'},\ldots,b_{s'+\ell-1})=1$}{\Return. }
			\kwLet $\ell'\gets\ell/2$; \\
			Sequentially execute $\mathcal{A}(\ell',s,s')$, $\mathcal{A}(\ell',s+\ell',s')$, 
			$\mathcal{A}(\ell',s,s'+\ell')$ and $\mathcal{A}(\ell',s+\ell',s+\ell')$.
		\end{algorithm}	
		
		It is easy to see that $\mathcal{A}(\ell,s,s')$ outputs all the collisions between the two intervals $a[s,\ldots,s+\ell-1]$ and $b[s',\ldots,s'+\ell-1]$, since whenever $\ell>1$ and there exists at least one collision (which is checked by the \textsc{ED} call), the algorithm splits each interval into two halves, and solve all four pairs of sub-intervals with the four recursive calls. Hence $\mathcal{A}(n,1,1)$ solves \textsc{SetCollision}.
		
		The space usage of $\mathcal{A}(n,1,1)$ is $\widetilde{O}(1)$, since there are $O(\log n)$ levels of recursion and each recursive call locally uses $\widetilde{O}(1)$ space. To bound the time usage, the key observation is that there are at most $n$ collisions. Therefore, although there could be as much as $(n/\ell)^2$ possible recursive calls to $\mathcal{A}$ at the level of recursion with interval length $\ell$, there are in fact at most $O(n)$ actual calls within each level, while the rest are prematurely stopped because of the \textsc{ED} check. Taking the summation over $\ell=2^t$ for $t=0,\ldots,\log n$, the total time usage of $\mathcal{A}(n,1,1)$ bounded by
		\[\sum_{t\leq \frac{1}{2}\log n}O(n)\cdot \widetilde{O}(2^t)+\sum_{t>\frac{1}{2}\log n}\left(\frac{n}{2^t}\right)^2\cdot\widetilde{O}(2^t)=\widetilde{O}(n^{3/2}).\]
		
		When the space $S$ is larger, in order to leverage the space advantage and reduce the time usage we need to \emph{parallelize} the algorithm $\mathcal{A}$. However, the core of algorithm $\mathcal{A}$ is the black-box \textsc{ED} algorithm, whose instances cannot be parallelized if they are highly adaptive. Therefore from now on, we assume that the space-$\widetilde{O}(1)$ and time-$\widetilde{O}(n)$ \textsc{ED} algorithm is oblivious.
		
		To understand how oblivious \textsc{ED} algorithm helps parallelization, consider the recursion level with $\ell=\sqrt{n}$. At this level, we need to answer $\mathrm{ED}(a_s,\ldots,a_{s+\ell-1},b_{s'},\ldots,b_{s'+\ell-1})$ for all $n$ pairs of $s,s'\in\{1,\sqrt{n}+1,\ldots,n-\sqrt{n}+1\}$. We can call the oblivious \textsc{ED} algorithm to solve the instance with $s=s'=1$, and call it again to solve another instance with $s=\sqrt{n}+1,s'=1$. Because the algorithm is oblivious, whenever $a_i$ (resp. $b_i$) is queried in the first instance, $a_{i+\sqrt{n}}$ (resp. $b_i$) is queried in the second instance at the exact same time step. That means the two algorithm instance can be interleaved, using double the space while the queries to $B$ do not need to be repeated. Take a step further, we can interleave the $4$ instances of \textsc{ED} with $s,s'\in\{1,\sqrt{n}+1\}$, using $4$ times the space but only \emph{double} the time.
		
		In our actual algorithm, we partition $\{1,\sqrt{n}+1,\ldots,n-\sqrt{n}+1\}$ into $\sqrt{n/S}$ groups, each of size $\sqrt{S}$. With the idea stated above, for each pair of groups of $s$ and $s'$, we can solve all the \textsc{ED} instances within this pair (there are $S$ instances) with space $\widetilde{O}(S)$ and time $\widetilde{O}(\sqrt{nS})$. As there are $n/S$ pairs of groups, the overall time usage all the \textsc{ED} instances at level $\ell=\sqrt{n}$ is $\widetilde{O}(\sqrt{n^3/S})$.
		
		More generally, using the same idea, we design a parallelized version of $\mathcal{A}$, which is the algorithm $\mathcal{A}^*(\ell, (s_i,s_i')_{i\in I})$ described as \Cref{alg3} below, that takes as an argument a list of $|I|\leq S$ pairs of $s$ and $s'$.
		
		\begin{algorithm}
			\DontPrintSemicolon
			\caption{Parallelized Divide-and-Conquer Algorithm $\mathcal{A}^*(\ell, (s_i,s_i')_{i\in I})$} \label{alg3}
			\If {$\ell=1$}{
				\For {$i\in I$}
					{\lIf {$a_{s_i}=b_{s_i'}$}{Output $(s_i,s_i',a_{s_i})$;}}
				\Return.
			}
			Solve $e_i\gets\mathrm{ED}(a_{s_i},\ldots,a_{s_i+\ell-1},b_{s_i'},\ldots,b_{s_i'+\ell-1})$ for all $i\in I$ in parallel; \\
			\kwLet $\ell'\gets\ell/2$, $Q\gets\varnothing$; \\
			\ForEach {$i\in I$ such that $e_i=0$}{
				\For {$(\Delta s,\Delta s')\gets (0,0),(\ell',0),(\ell',0),(\ell',\ell')$}{
					Add $(s_i+\Delta s,s_i'+\Delta s')$ to the queue $Q$; \\
					\If {$|Q|=S$ or reaching the end of the algorithm}{
						Execute $\mathcal{A}^*(\ell',Q)$; \\
						$Q\gets\varnothing$.
					}
				}
			}		
		\end{algorithm}
		
		It is clear from the description that $\mathcal{A}^*(\ell, (s_i,s_i')_{i\in I})$ functions the same as the sequential execution of $\mathcal{A}(\ell,s_i,s_i')$ for all $i\in I$. Our final algorithm for \textsc{SetCollision} is to run sequentially $\mathcal{A}^*(\sqrt{n},G\times G')$, for all $G$ and $G'$ chosen from the $\sqrt{n/S}$ groups of size $\sqrt{S}$ that partitions $\{1,\sqrt{n}+1,\ldots,n-\sqrt{n}+1\}$, and thus it correctly outputs all collisions between set $A$ and $B$. Each recursive call of $\mathcal{A}^*$ uses $O(S\log n)$ space locally, plus the $\widetilde{O}(S)$ space to compute at most $S$ instances of \textsc{ED} in parallel. As there are $O(\log n)$ levels of recursion, the overall space usage is $\widetilde{O}(S)$.
		
		To bound the time usage, we first examine how much time is used to solve $S$ instances of \textsc{ED} in parallel. Fix the initial argument $G$ and $G'$ at the start of the recursion and focus on one level of recursion with interval length $\ell$. At this level, one instance of the \textsc{ED} algorithm takes $\widetilde{O}(\ell)$ time. Since the input intervals for these \textsc{ED} instances are either the same or disjoint, each query is repeated for at most $|G|\cdot \sqrt{n}/\ell$ times at its parallel places after the interleaving parallelization. Thus the time usage for solving \textsc{ED} is $\widetilde{O}(|G|\cdot \sqrt{n})=\widetilde{O}(\sqrt{nS})$. As the rest of steps take $O(S)\leq O(\sqrt{nS})$ time, altogether each recursive call of $\mathcal{A}^*$ locally takes $\widetilde{O}(\sqrt{nS})$ time, regardless of the level of recursion.
		
		On the other hand, let us call a recursive call $\mathcal{A}^*(\ell, (s_i,s_i')_{i\in I})$ \emph{complete} if $|I|=S$, and \emph{incomplete} if $|I|<S$. Since there are at most $n$ collisions, at each level of the recursion there are at most $O(n/S)$ complete calls, while each call produces at most one incomplete call in the next level. Initially there are $n/S$ calls, and therefore the total number of calls to $\mathcal{A}^*$ in our final algorithm is $\widetilde{O}(n/S)$. So the total running time is $\widetilde{O}(\sqrt{n^3/S})$.
	\end{proof}
	
	Since \textsc{SetCollision} has the randomized lower bound $T^2S=\widetilde{\Omega}(n^3)$, \Cref{thm:ed} implies that any polynomial separation between randomized and deterministic time-space tradeoffs of \textsc{SetCollision} (or its variants such as \textsc{SetIntersection}) would answer \Cref{openproblem} on \textsc{ElementDistinctness}.

	Notice that in the reduction of \Cref{thm:ed}, the input guarantee that both lists $A$ and $B$ are sets is only used so that \textsc{ED} decides the distinctness between the two lists. Without the guarantee, we can instead resort to the \textsc{ListDistinctness} problem studied in \cite{BGN018}.
	
	\begin{definition}
		In the \textsc{ListDistinctness} problem (\textsc{LD} for short), the input contains two unordered lists $(a_1,\ldots,a_n)$ and $(b_1,\ldots,b_n)$ from a fixed domain $D$, with $|D|=\poly(n)$. The output is $1$ if there exist $i,j\in[n]$ such that $a_i=b_j$, and $0$ otherwise.
	\end{definition}
	
	\textsc{LD} is at least as harder as \textsc{ED}, and while \textsc{ED} can be solved in $\widetilde{O}(1)$ space and $\widetilde{O}(n^{3/2})$ time, no algorithm even with $n^{o(1)}$ space and $n^{2-\Omega(1)}$ time was known for \textsc{LD}. The proof of \Cref{thm:ed} can be altered to show that the problem of \textsc{$n$-Collision} reduces deterministically to \textsc{LD}:
	
	\begin{definition}\label{def:nc}
		In the \textsc{$n$-Collision} problem, the input is an unordered lists $(a_1,\ldots,a_n)$ of elements in $D$. The output consists of $n$ distinct collisions, that are triples $(i,j,x)$ such that $i\neq j$ and $a_i=a_j=x$, or all of the collisions if there are less than $n$ of them.
	\end{definition}

	Strictly speaking, the \textsc{$n$-Collision} problem is not a function, but rather a relational problem, as the collection of outputted collisions is not uniquely determined. However, a time-space lower bound of $T^2S=\widetilde{\Omega}(n^3)$ is still know for \textsc{$n$-Collision} \cite{Dinur20}.
	
	\begin{theorem}\label{thm:ld}
		If \textsc{LD} can be solved deterministically with space $\widetilde{O}(1)$ and time $\widetilde{O}(n)$, then \textsc{$n$\nobreakdash-Collision} can be solved deterministically with space $\widetilde{O}(1)$ and time $\widetilde{O}(n^{3/2})$. Furthermore, if the algorithm for \textsc{LD} is oblivious, then for every $S\leq n$, there is a deterministic (non-oblivious) algorithm that solves \textsc{$n$-Collision} with space $\widetilde{O}(S)$ and time $\widetilde{O}(\sqrt{n^3/S})$.
	\end{theorem}
	\begin{proof}
		Notice that the collisions found in the algorithms in \Cref{thm:ed} are all distinct. By setting a global counter for the number of collisions already found and outputted, the algorithms and proofs in \Cref{thm:ed} can be copied verbatim to show a reduction to \textsc{LD} from the problem \textsc{$k$\nobreakdash-ListCollision}, where the input consists of two unordered lists of size $n$ that may contain duplicates, and the output contains $k$ collisions (if exist) between the two lists. If \textsc{LD} can be solved deterministically with space $\widetilde{O}(1)$ and time $\widetilde{O}(n)$, then the deterministic algorithm for \textsc{$k$\nobreakdash-ListCollision} works in space $\widetilde{O}(S)$ and time $\max\{m,n\}\cdot\widetilde{O}(\sqrt{n/S})$, where $m$ is the actual number of collisions outputted ($S$ can be arbitrary when the algorithm for \textsc{LD} is oblivious, and $S=O(1)$ in the general case).
		
		Now notice that the complete graph over $n$ vertices can be partitioned into a set of complete bipartite graphs, $2^{t-1}$ of which being of size $(n/2^t,n/2^t)$ for $t=1,\ldots,\log n$. We apply \textsc{$n$\nobreakdash-ListCollision} on each pairs of lists of size $n/2^t$ defined by these bipartite graphs, until $n$ collisions are found. This clearly solves the \textsc{$n$-Collision} problem with space $\widetilde{O}(S)$. Suppose that the number of collisions actually outputted on each bipartite graph is $m_1,m_2,\ldots$ respectively, then the total time usage is
		\begin{align*}
			&\max\left\{m_1,\frac{n}{2}\right\}\cdot\widetilde{O}\left(\sqrt{\frac{n}{2S}}\right)+
			\max\left\{m_2,\frac{n}{4}\right\}\cdot\widetilde{O}\left(\sqrt{\frac{n}{4S}}\right)+
			\max\left\{m_3,\frac{n}{4}\right\}\cdot\widetilde{O}\left(\sqrt{\frac{n}{4S}}\right)+\cdots \\
			\leq\  &(m_1+m_2+\cdots)\cdot\widetilde{O}\left(\sqrt{\frac{n}{2S}}\right)+
			\sum_{t=1}^{\log n}2^{t-1}\cdot \frac{n}{2^t}\cdot \widetilde{O}\left(\sqrt{\frac{n}{2^tS}}\right)
			\quad =\quad \widetilde{O}\left(\sqrt{n^3/S}\right). \qedhere
		\end{align*}
	\end{proof}

	Similarly, we have the corollary of \Cref{thm:ld} that any polynomial separation between randomized and deterministic time-space tradeoffs of \textsc{$n$-Collision} (or its variants such as \textsc{MemoryGame} \cite{chakrabarti2017time}) would answer \Cref{openproblem} on \textsc{ListDistinctness}.

	\bibliography{main}
	\bibliographystyle{alpha}
	
\end{document}